\documentclass[letterpaper]{article}
\usepackage[cmex10]{amsmath}
\usepackage{amsthm}
\usepackage{enumerate}
\usepackage{graphicx}
\usepackage{epstopdf}

\usepackage{amsfonts,amscd}
\usepackage{float}

\newtheorem{theorem}{Theorem}[section]

\newtheorem{proposition}[theorem]{Proposition}

\newcommand{\be}{\begin{equation}}
\newcommand{\ee}{\end{equation}}
\newcommand{\bea}{\begin{eqnarray}}
\newcommand{\eea}{\end{eqnarray}}
\newcommand{\beann}{\begin{eqnarray*}}
\newcommand{\eeann}{\end{eqnarray*}}
\newcommand{\ip}[2]{\langle{#1}\mid{#2}\rangle}

\newcommand{\unity}{{1\hskip -3pt \rm{I}}}
\begin{document}
\title{An efficient algorithm for the entropy rate of a hidden Markov model with unambiguous symbols}
\author{
    Jaideep Mulherkar\\
    Dhirubhai Ambani Institute of \\
    Information and Communication Technology\\
    jaideep\_mulherkar@daiict.ac.in
}
\maketitle
\begin{abstract}
We demonstrate an efficient formula to compute the entropy rate $H(\mu)$ of a hidden Markov process with $q$ output symbols where at least one symbol is unambiguously received. Using an approximation to $H(\mu)$ to the first $N$ terms we give  a $O(Nq^3$) algorithm to compute the entropy rate of the hidden Markov model. We use the algorithm to estimate the entropy rate  when the parameters of the hidden Markov model are unknown. In the case of $q =2$ the process is the output of the Z-channel and we use this fact to give bounds on the capacity of the Gilbert channel.
\end{abstract}
\section{Introduction}
Entropy rate of a stationary stochastic process $\{X_n\}_{n=0}^{\infty}$ is the limit
\bea
\label{eq:entratedefn}
H(\mu) = \lim_{n\rightarrow\infty} \frac{S_n(X_1,X_2,...,X_n)}{n}
\eea
where $\mu$ is the measure associated with the process and $S_n(X_1,X_2,...,X_n)$ is the joint entropy of $\{X_1,X_2,...,X_n\}$. It amounts to the average amount of information per symbol.  In this paper we study the entropy rate of a hidden Markov process (HMP) that has at least one unambiguous received  symbol. A received symbol is unambiguous if after receiving the symbol one can conclude with certainty the state or input symbol. An example of an HMP  with one unambiguous received symbol is the output of the Z-channel with Markov input process which has been used to model optical communication systems. A closed form formula exists when the process is Markov however a tractable formula for the entropy rate of a general HMP is still an outstanding problem. Entropy rate of a HMP was first studied by Blackwell in 1957 \cite{DB1957}. Blackwell showed that the entropy rate of a HMP can be computed as an integral of a function defined on the simplex with respect to a measure. Unfortunately in most cases the measure is quite complicated and computation of the entropy rate using this method us not feasible. Birch \cite{JB1962} showed that the entropy rate can be upper and lower bounded by functions that converge exponentially fast to the entropy rate. A formula for the entropy rate of a HMP also assumes importance because of the use of hidden Markov Models in practical applications such as speech and image processing, bioinformatics and communication and information theory \cite{R1989}. Recently there has been a renewed interest in computing the entropy rate. Entropy rate calculations based on ideas from filtering theory have been done\cite{OW2006}, connections of entropy rate to Lyapunov exponents of random matrices have been studied in \cite{JSZ2004,HGG2006}, connections with statistical mechanics in \cite{ZDKA2006, AA2008}, and in capacity calculations of finite state channels in \cite{HP2003}. In this paper we follow the approach of algebraic measures \cite{FNS1992}. Algebraic measures were introduced by Fannes, Nachtergaele and Werner in the context of quantum spin systems as classical analogues of finitely correlated states  and were shown to be in one to one correspondence with functions of Markov processes or hidden Markov processes. In \cite{MMN2012} we used the approach of algebraic measures to compute the entropy rate of a hidden Markov model with at least one unambiguous symbol and showed that an approximation to the formula converges exponentially fast to the entropy rate. Our paper is organized as follows; in section \ref{sec:Background} we give background about the entropy rate problem, introduce the noise model and review the results of \cite{MMN2012}, in section \ref{sec:EntComputation} we show an efficient algorithm to compute the entropy rate and present numerical estimates of the entropy rate using a sequence of observed symbols and in section \ref{sec:Gilbert} we use the results to derive bounds on the capacity of the Gilbert channel.

\section{Background}
\label{sec:Background}
\subsection{Setup}
Consider a stationary Markov process $\{X_1,X_2,...\}$ taking values in an alphabet $N=\{0,1,...,k-1\}$. Let $E$ be the transition matrix and $\nu$ be the stationary Markov measure associated to the process. Let $F_a \in M_k(\mathbb{R}) (k\times k$ matrices with entries in $\mathbb{R}$) be the matrix with the only non-zero row to be the $a^{th}$ row of the transition matrix $E$, that is 
\bea
\label{eq:MarkovMatrix}
(F_a)_{b,c} = \delta_{a,b} \frac{\nu((b,c))}{\nu((b))}
\eea
so that $E = \sum_{a\in N} F_a$. Let $\unity \in \mathbb{R}^k$ be the vector with all components equal to $1$ and $\tau \in \mathbb{R}^k$ be such that $\tau_a = \nu((a))$, the $a^{th}$ component of the stationary distribution. The Markov measure $\nu$ can be represented in terms of a triplet $( \tau,\unity,(F_a)_{a\in N},)$. It is easy to verify that
\bea
\label{eq:HMP 1} 
\nu((\omega_1,...,\omega_n)) &=&\ip{\tau}{F_{\omega_1}...F_{\omega_n} \unity}
\eea
where $\ip{u}{v}= u^{T}v$ is the usual inner product on $\mathbb{R}^k$. Let  $\{Y_1,Y_2,...\}$ with $Y_i \in K=\{0,1,...,q-1\}$ be the hidden Markov process resulting from a noisy observation of the Markov process given by the matrix  $R = [r_{ab}]$ with $r_{ab}= Pr[Y_i=a|X_i=b]$. One can view the output $\{Y_n\}$ as a Markov source $\{X_n\}$ through a discrete memoryless channel. The noisy observation of the Markov process induces a translation invariant measure $\mu$ on $K^{\mathbb{Z}}$ which can be written as
\bea
\mu(\epsilon_1,\epsilon_2,...,\epsilon_n) = \sum_{\stackrel{\omega_1,\omega_2,...,\omega_n}{\omega_i \in N}}  r_{\epsilon_n\omega_n}r_{\epsilon_{n-1}\omega_{n-1}}\cdots r_{\epsilon_1\omega_1}\nu(\omega_n|\omega_{n-1})...\nu(\omega_2|\omega_1)\nu(\omega_1)
\eea
The hidden Markov process can be equivalently be represented by a function $\Phi:N\rightarrow K$ and the measure $\mu$ associated with  can be written as
\bea
\mu(\epsilon_1,\epsilon_2,...,\epsilon_n) = \sum_{\stackrel{\omega_1,\omega_2,...,\omega_n}{\phi(\epsilon_i)= \omega_i}} \nu(\omega_1,\omega_2,...,\omega_n)
\eea
We can also represent the hidden Markov process in terms to a triplet. Let
\bea
\label{eq:HMP 2}
E_a = \sum_{b\in L}r_{ab} F_b
\eea
It can be checked that the measure $\mu$ can be generated by triplet $(\tau,\unity,(E_a)_{a\in K})$ so that
\bea
\label{eq:measure formula}
\mu((w_m,...,w_n))=\ip{\tau}{E_{w_m}...E_{w_n} \unity}
\eea
Translation invariant measures on $K^{\mathbb{Z}}$ which can be represented in terms of triplets were termed as manifestly positive algebraic measures in \cite{FNS1992} and they were shown to be in one to one correspondence with functions of Markov processes or hidden Markov processes.

There is a well known formula for the entropy rate of the the Markov measure $\nu$. We can write the
\bea
H(\nu) = \sum_{a,b} \nu((a))E_{a,b}
\eea
A tractable formula for the entropy rate of a hidden Markov process is still an open and challenging problem. Blackwell  was the first to study the entropy rate of a hidden Markov process. He showed in \cite{DB1957} that the entropy rate of a hidden Markov process can be written as an integral of a function on a simplex with respect to a measure on the simplex. The entropy rate given by Blackwells formula is 
\bea
\label{eq:Entropy rate integral formula}
H(\mu)&=& \sum_{a\in K} \int_{\mathcal{W}} h_a(w)\phi(dw)
\eea
and $\phi(dw)$ is a probability measure on the simplex $\mathcal{W}=\{(w_1,w_2,...,w_N) | \sum_i w_i =1\}$ and $h_a$ is some function on the simplex.
However, practically computing the entropy rate of a hidden Markov process using the Blackwell formula is difficult since the Blackwell measure can be hard to evaluate. Birch \cite{JB1962} showed that the monotonically decreasing sequence $G_n = S(Y_n|Y_{n-1},Y_{n-2},...,Y_1)$ converges exponentially fast to the entropy rate, that is, there exist positive constants $M$ and $0< \rho <1$ such that
\bea
\label{eq:Birch}
G_n - H(\mu) \leq M\rho^{n-1}
\eea
It can be seen that
\bea
\label{eq:entropyratedefn2}
G_n = S(Y_n,Y_{n-1},...,Y_1) - S(Y_{n-1},Y_{n-2}...,Y_1)
\eea
One can compute the entropy rate using the equation \ref{eq:entropyratedefn2} but it is clear that computing the entropy rate using this formula by calculating the joint probabilities involved will take time that is exponential in $n$.

\subsection{Noise model and formula for the entropy rate}
In \cite{MMN2012} we considered a specific noise model which we call a hidden Markov model with at least one unambiguous received symbol. If the symbol $0$ is transmitted then it is always received as $0$ at the other end.  On the other hand if any of the other symbol is transmitted then it is either received without any error or received as the symbol $0$ with a small error probability. That is $P(Y_i=0|X_i=0) =1$, $P(Y_i=0|X_i=a) = \epsilon_a$ and $P(Y_i=a|X_i=a) = 1- \epsilon_a$ for $a=1,...,q-1$ and $P(Y_i=b|X_i=a)=0$ when $0\neq b \neq a$. Here we consider the symbols $1,2,...,q-1$ to be unambiguous, since if any one of them is received then that same symbol must have been transmitted. For $q=2$ this model is the familiar Z-channel. See figure \ref{fig:Model} for a description of the model in the case $q=2$ and $q=3$. 
\begin{figure}[ht]
\begin{center}
\includegraphics[height=30 mm,width=70 mm]{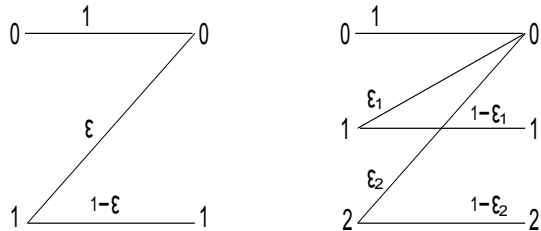}
\caption{The noise model for $q=2$ and $q=3$. For $q=2$ this noise model results in the familiar Z-channel which has been used as a model for transmission problems in optical communications. $1$ and $2$ are the unambiguous symbols for $q=3$ since if either a $1$ or a $2$ was received then we can conclude with certainty that the sent symbol was the same. If a $0$ is received then all of the three symbols could have been transmitted; $1$ and $2$ with probability $\epsilon_1$ and $\epsilon_2$ and 0 with probability $1-\epsilon_1 -\epsilon_2$.}
\label{fig:Model}
\end{center}
\end{figure}
Let the matrices $\{F_a\}$ be the matrices that describe the uncorrupted Markov source as in equation (\ref{eq:HMP 1}).
For this noise model we write the matrices $\{E_a\}$ given by equation (\ref{eq:HMP 2}) as
\bea
\label{eq:Matrix E_a defn}
E_0 &=& F_0 + \sum_{a=1}^{q-1} \epsilon_a F_a\\ \nonumber
E_a &=& (1-\epsilon_a) F_a \quad \text{for} \,a= 1,...,q-1\\ \nonumber
\sum_{a\in K} E_a &=& E
\eea
Let $\Gamma_a:\mathcal{W}\rightarrow \mathcal{W}$ be a mapping on the simplex $\mathcal{W}$ defined by
\bea
\label{eq:Gamma defn}
\Gamma_a(\nu) = \frac{E_a^T\nu}{\ip{\nu E_a}{\unity}}
\eea
Let $e_i, i=0,1...,q-1$ denote the transpose of the $(i+1)^{st}$ row of $E=[e_{ij}]$. In \cite{MMN2012} we showed that the support of the Blackwell measure for the hidden Markov model described by the noise model in this section is countable.
\begin{proposition}[\cite{MMN2012}]
\label{prop:support of phi}
For the HMP with one or more unambiguous received symbol the support of the measure $\phi$ is given by
\bea
\Delta = \overline{\{\Gamma_0^m e_j | j \in \{1,..,q-1\} ; m \in \mathbb{N}_0\}}
\eea
\end{proposition}
Next we state the assumptions and statement of the main theorem from \cite{MMN2012} for the entropy rate of the hidden Markov process under consideration. Let $p = \min_{ij} e_{ij}$ and $P=\max_{ij} e_{ij}$.
\beann
\textit{Assumption } 1:&&\\
i)&&0<p \le P < 1, \epsilon_0 = 1, \, \epsilon_a > 0 \quad \forall a \in \{1,...,q-1\}\\
ii)&& E_0 \quad \text{is a one to one mapping}
\eeann
Define
\bea
\label{eq:c_jm}
c_{j,m} = \prod_{i=1}^m \ip{\Gamma_0^{m-i} e_j}{E_0 \unity}
\eea
Let $A$ be the $q\times q-1$ matrix defined by entries.
\bea
\label{eq:A matrix 1}
A_{ij} &=& -\delta_{ij} + \sum_{m=0}^{\infty}\ip{\Gamma_0^m e_j}{E_i \unity}c_{j,m} \quad \text{if}\, i\neq q, q\neq 2 \nonumber \\ 
A_{ij} &=& 0 \qquad \text{if} \, i\neq q, q= 2 \nonumber \\
A_{qj} &=& \sum_{m=0}^{\infty} c_{j,m} 
\eea.
\beann
\label{eq:Phi and b}
\Phi = [\phi(e_1)\cdots\phi(e_{q-1})]^{T} \in \mathbb{R}^{q-1}, b  = [0\, 0\cdots1]^{T} \in \mathbb{R}^q
\eeann
Here $\phi(e_i)$ is the weight of measure $\phi$ at the point $e_i \in \mathbb{R}^q$.
Let $h_a:\mathcal{W}\rightarrow \mathbb{R}$ be the function defined as
\bea
h_a(\nu) = -\ip{\nu}{E_a \unity}\log\ip{\nu}{E_a \unity}
\eea
\begin{theorem}[\cite{MMN2012}]
\label{thm:Entropy rate}
Under \textit{Assumption} 1  the entropy rate of the measure $\mu$ associated with the hidden Markov process with the noise model described in this section is given by
\bea
\label{eq:Entropy rate formula}
H(\mu) = \sum_{j=1}^{q-1}\sum_{m=0}^{\infty}\sum_{a=0}^{q-1} h_a(\Gamma_0^m e_j) c_{j,m} \Phi_j 
\eea
\end{theorem}
In \cite{MMN2012} we showed that an approximation to the formula for $H(\mu)$ converges exponentially fast to the entropy rate. For the HMP under consideration the result for the exponential convergence was much more simpler to show than Birch's general result given by equation \ref{eq:Birch}. For the approximation to $H(\mu)$ let 
\bea 
\label{eq:A matrix approximation}
A&=&\hat{A} + R
\eea
where the entries of $R$ are the tails ($(N+1)^{st}$ term onwards) of the entries of $A$. 
Let $\hat{\Phi}$ be the least square solution to 
\bea
\label{eq:Linear system}
\hat{A} \Phi = b
\eea
\bea
\label{eq:Least squares solution}
\text{Therefore} \, \hat{\Phi} = \hat{A}^{\dagger} b
\eea
where $\hat{A}^{\dagger} = (A^TA)^{-1}A^T$ is the pseudo-inverse of $A$. Define
\beann
H_N(\mu)= \sum_{j=1}^{q-1} \sum_{m=0}^N\sum_{a=0}^{q-1} h_a(\Gamma_0^m e_j) c_{j,m} \hat{\Phi}_j
\eeann
\bea
\text{and} \,\gamma   := \max_{j} \sup_{k} \sum_{a=0}^{q-1} \epsilon_a [\Gamma_0^k e_j]_a 
\eea
 We have the following theorem 
\begin{theorem}[\cite{MMN2012}]
\label{thm:Entropy rate appx}
Under \textit{Assumption} 1  the entropy rate $H(\mu)$ of the hidden Markov process with the noise described can be approximated to $O(\gamma^{N+1})$ by $H_N(\mu)$ and we have 
\bea
\label{eqn:Entropyformula}
|H(\mu)-H_N(\mu)| &\leq& B\gamma^{N+1} \quad \text{with} \,\, B = \frac{q}{1-\gamma}\big(1+\frac{q\|\hat{A}^{\dagger}\|_1}{1-\gamma}\big)
\eea
\end{theorem}

\section{Estimation and computation of entropy rate} 
\label{sec:EntComputation}
\subsection{An efficient algorithm to compute entropy rate}
Birch's result (equation \ref{eq:Birch}) shows that the computation of the entropy rate of a general hidden Markov chain using the monotonically decreasing sequence $G_n$ converges exponentially to the actual entropy rate. However the computation of $G_n$ using equation \ref{eq:entropyratedefn2} method takes exponential time in n. In this section we give a $O(Nq^3)$ algorithm to compute the entropy rate of the hidden Markov model with unambiguous symbols using the approximate formula $H_N(\mu)$. If $\delta = |H(\mu) -H_N(\mu)|$ is accuracy with which we compute $H(\mu)$ then we get an algorithm that is $O(\log\frac{1}{\delta})$ in terms of the accuracy as compared to $O(\frac{1}{\delta})$ if we use the brute force formula $G_n$. In this section we prove these results and substantiate them with numerical computations. The algorithm to compute the entropy rate is as follows:\\
\textbf{\underline{Algorithm I:}\\}
\textbf{Inputs}:\\
i)A $q\times q$ transition matrix  $E$ of the Markov chain $E_{ij} = P(X_n=j|X_{n-1} =i)$.
ii)A $q\times q$ channel probability matrix $R$ with $R_{ab} = P(Y_n=a|X_n=b)$ according to the hidden Markov process under consideration.\\
iii) $N$ the number of terms of the approximate formula.\\
Both $E$ and $R$ should satisfy conditions specified by Assumption 1.\\
\textbf{Output}: The entropy rate $H$ of the hidden Markov model.\\
\textit{\textbf{Step 1:}}\\
From the matrices $E$ and $R$ construct matrices $E_0$ according to equation \ref{eq:Matrix E_a defn}.\\
For $j=[0\cdots q-1\}$ and $m=[1\cdots N]$  compute $\Gamma_0^m e_j$. (where $\Gamma_0$ is given by equation \ref{eq:Gamma defn} and $e_j$ is the transpose of the $(i+1)^{st}$ row of $E$ )\\
\textit{\textbf{Step 2:}}\\
For $j=[0\cdots q-1]$ and $m=[1\cdots N]$ compute $c_{j,m} = \Pi_{i=1}^m\ip{\Gamma_0^{m-i}e_j}{E_0\unity}$ .\\
\textit{\textbf{Step 3:}}\\
Compute the entries of the matrix $\hat{A}$ given by equation \ref{eq:A matrix approximation}, the pseudo-inverse $\hat{A}^{\dagger}=(A^TA)^{-1}A^T$ and then vector $\hat{\Phi} =\hat{A}^{\dagger}b$.\\
\textit{\textbf{Step 4:}}\\
Using the precalculated values of $\Gamma_0^m e_j$, $c_{j,m}$ and $\Phi_j$ in Steps 1,2 and 3 do the following computation.\\
$H=0$.\\
For $j \in [1\cdots q-1]$, $m \in [0\cdots N]$ and $a  \in [0\cdots q-1]$\\
$\quad H = H + -(\Gamma_0^m e_j)_a \log (\Gamma_0^m e_j)_a c_{j,m}\hat{\Phi}_j$.\\
Output  $H = H_N(\mu)$ as the entropy rate.
\begin{theorem}
\label{thm:timecomplexity1}
Run time complexity of Algorithm I to compute the entropy rate is $O(Nq^3)$
\end{theorem}
\begin{proof}We analyze the steps of Algorithm I\\
\begin{itemize}
\item
Each computation in \textit{Step 1} is matrix multiplication of a $q\times q$ matrix $\Gamma_0$ with a $q\times 1$ vector $e_j$ which using a standard matrix multiplication algorithm requires $O(q^2)$ time. There are $Nq$ total such computations and hence the time complexity of \textit{Step 1} is $O(Nq^3)$. 
\item
In \textit{Step 2} one requires the computation of $c_{j,m} = \Pi_{i=1}^m\ip{\Gamma_0^{m-i}e_j}{E_0\unity}$. $E_0\unity$ takes $O(q)$ time and for each $j$ the inner product $\ip{\Gamma_0^k e_j}{E_0\unity}$ can be done in $O(q)$ time. $c_{j,k}$ can be computed iteratively as $c_{j,k} = c_{j,k-1}\ip{\Gamma_0^k}{E_0\unity}$ and since there are $Nq$ such computations total time taken by \textit{Step 2} is $O(Nq^2)$.
\item
In \textit{Step 3} we first compute the matrix elements of the $q\times q-1$ matrix $\hat{A}$. Each term of $\hat{A}$ is given by equation \ref{eq:A matrix approximation} up till the first $N$ terms. Each matrix entry thus requires $O(Nq)$ time and since there are order $q^2$ terms computing $\hat{A}$ requires $O(Nq^3)$ time. Next we compute the pseudo inverse $\hat{A}^{\dagger}=(A^T A)^{-1}A^T$ which is a combination of matrix multiplication and taking inverse which by standard methods takes $O(q^3)$ time. Computing $\Phi =\hat{A}^{\dagger}b$ requires $O(q^2)$ time, hence the total time required in \textit{Step 3} is $O(Nq^3)$.
\item
Finally \textit{Step 4} has $Nq^2$ basic operations of addition or multiplication and hence requires $O(Nq^2)$ time.
\end{itemize}
From the above analysis we get that the time complexity of Algorithm I is $O(Nq^3)$.
\end{proof}
\begin{theorem}
The running time of Algorithm I to compute $H_N(\mu)$ to within $\delta$ accuracy of $H(\mu)$ is $O(\log\frac{1}{\delta})$
\end{theorem}
\begin{proof}
From the bound of equation \ref{thm:Entropy rate appx} we get that 
\beann
|H(\mu)-H_N(\mu)| &\leq& B\gamma^{N+1} \quad \text{with} \,\, B = \frac{q}{1-\gamma}\big(1+\frac{q\|\hat{A}^{\dagger}\|_1}{1-\gamma}\big)
\eeann
Therefore to obtain a $\delta$ accuracy in computation of $H(\mu)$ we need
\beann
\delta &\leq& B\gamma^{N+1} \qquad \text{that is}\\
\frac{1}{\delta} &\geq& B\gamma^{N+1}\\
\log\big(\frac{1}{\delta}\big) &\geq& \log B + (N+1)\log \gamma\\
\eeann
dividing by the negative quantity $\log \gamma $ gives
\bea
\label{eq:appxNterms}
N+1 \geq \frac{\log \big(\frac{1}{\delta B}\big)} {\log\gamma}
\eea
Combining with theorem \ref{thm:timecomplexity1} we get that the time complexity of Algorithm I to compute $H_N(\mu)$ to $\delta$ accuracy is $O(\log\frac{1}{\delta})$.
\end{proof}
We present a numerical example for approximating the entropy rate formulas given by theorem \ref{thm:Entropy rate appx}.
Let $q=3$, $\epsilon_1 = 0.01$ and $\epsilon_2 = 0.02$. The transition matrix we use is
\beann
E &=& \begin{pmatrix} 0.4 & 0.25 & 0.35 \\ 0.25 & 0.45 & 0.3\\ 0.2 & 0.55 & 0.25 \end{pmatrix}
\eeann
The results of the entropy rate calculations are seen in table \ref{table:entropy rate 1}. A comparison with calculations of the entropy rate for the same HMP but by using the brute force formula of equation \ref{eq:entratedefn} is seen in table \ref{table:entropy rate 2}.
\begin{table}[ht]
\begin{center}
\begin{tabular}{|c|c|c|c|}	
\hline
N       &   $H_N(\mu)$       & err(N) bound            & Execution time (in secs)\\
\hline
10      &   1.520946691296695 & 0.3561                    &  0.0077     \\
\hline
20      &   1.520947864830033 & 0.0030                    &  0.0129  \\
\hline   
30      &   1.520947864969799 & $2.6758 \times 10^{-5}$   &  0.0197          \\
\hline
40      &   1.520947864969815 & $2.3193 \times 10^{-7}$   &  0.0278\\
\hline
50      &   1.520947864969815 & $2.0103 \times 10^{-9}$    & 0.0289\\  
\hline
\end{tabular}
\caption{The estimated entropy rate $H_N(\mu)$ using
the formula given by theorem \ref{thm:Entropy rate appx}.}
\label{table:entropy rate 1}
\end{center}
\end{table}
%\begin{figure}[ht]
%\begin{center}
%\includegraphics[height=60 mm,width=80 mm]{Enterr.eps}
%\caption{Plot of the error $\text{err}(N)$ versus N.}
%\label{fig:Enterr3d}
%\end{center}
%\end{figure}
\begin{table}[ht]
\begin{center}
\begin{tabular}{|c|c|c|}	
\hline
n       &   $S_n(\mu)- S_{n-1}(\mu)$                    & Execution time (in secs)\\
\hline
5      &   1.520946036478195                            &  0.0581     \\
\hline
6      &   1.520947599473784                            &  0.127     \\
\hline   
7      &   1.520947829277763                           &  0.342      \\
\hline
8      &   1.520947860073111                           &  1.08\\
\hline
9      &   1.520947864301537                           & 3.479\\  
\hline
10     &  1.520947864877943                           &  11.14\\  
\hline
\end{tabular}
\caption{The estimated entropy rate using the brute force formula $H(\mu)=S_n(\mu) -S_{n-1}(\mu)$.}
\label{table:entropy rate 2}
\end{center}
\end{table}

\subsection{Estimating the entropy rate from an observed sequence}
In the previous subsection we have assumed that the transition matrix $E$ of the Markov chain and the noise parameters $\epsilon_a$ are known. However in many practical applications this is not the case. In this section we assume that we are only given an observation sequence and we have to estimate the entropy rate. In this method we use parameter estimation to estimate the transition matrix and the noise parameters $\epsilon_a$ and then use Algorithm I to compute the entropy rate. Let $Y$ be a vector of the observed symbols for time $t=1$ to $t=N$ and let $X$ be the corresponding hidden or state symbols and let $Z = \begin{pmatrix} Y\\ X \end{pmatrix}$. Let the unknown parameters be represented by $\theta = \begin{pmatrix} \eta \\ \epsilon\end{pmatrix}$ where $\eta$ is a vector containing the unknown transition matrix entries the and $\epsilon$ is the vector representing the noise parameters. We have 
\beann
p(Z| \theta) &=& p(Y,X|\theta) = p(Y|X,\epsilon)p(X|\eta)             
\eeann
Assuming the initial distribution of the Markov chain is uniform we get
\beann
p(Z| \theta)&=& \frac{1}{q}\Pi_{t=1}^N p(Y(t)|X(t), \epsilon) \eta_{t,t+1}
\eeann
and the log-likelihood function $L(Z|\theta)$ and the complete likelihood function $Q(\theta | \theta')$ respectively 
\bea
\label{eq:CLEq}
L(Z|\theta) &=& \log p(Z|\theta) = \sum_{t=1}^N p(Y(t)|X(t), \epsilon) + \log \eta_{t,t+1} - \log q \nonumber \\
Q(\theta | \theta')  &=& \sum_{X\in K^n} L(Z|\theta) p(X|Y,\theta')
\eea
To compute the complete likelihood function the $Q(\theta | \theta')$ conditional probabilities $p(X|Y,\theta')$ need to be estimated. This can be done using the Baum-Welsh forward-backward algorithm \cite{R1989}. Since the Markov input sequence it is only required to estimate the probabilities 
\beann
p(X(t)= k,X(t+1) = l|Y)= \frac{p(X(t)= k,X(t+1) = l,Y)}{P(Y)}
\eeann
For each $t$ we define row vector $1\times q$ vector $\alpha(t)$, $q\times 1$ vector $\beta(t)$ and $q\times q$ matrix $m(t)$
\beann
\alpha_k(t)  &=&  p(X(t)=   k , Y_1^{t-1}=y_1^{t-1})\\
\beta_t(l+1) &=& p(Y_{t+1}^N | X(t+1) = l)\\ 
m_{kl}       &=& p(Y(t)|X(t)=k)P(X(t+1)=l|X(t)=k)
\eeann
where we use the notation $p(Y_1^N) = p(Y_1,...,Y_N)$, then one can see that
\bea
\label{eq:Relalpambeta}
P(X(t)=k,X(t+1)=l,Y) = \alpha_k m_{kl}(t) \beta_l(t+1)
\eea
and we observe the following forward and backward recursion equations 
\bea
\label{eq:FBrec}
\alpha(t+1) &=& \alpha(t) m(t) \\\nonumber
\beta(t) &=& m(t)\beta(t+1) 
\eea
if $\gamma(t) = \alpha(t)\beta(t)$ then we have
\beann
\gamma(t) = \alpha(t)m(t)\beta(t+1) = \alpha(t+1)\beta(t+1) = \gamma(t+1)
\eeann
that is $\gamma(t)$ is time invariant and in fact
\beann
p(Y=y) = \sum_{k\in K} \sum_{l\in K} p(X(t)=k,X(t+1) = l,Y=y)
\eeann 
So that from equations \ref{eq:Relalpambeta} and recursion equations \ref{eq:FBrec} we get
\beann
p(Y=y) &=& \sum_{k\in K} \sum_{l\in K} p(X(t)=k,X(t+1) = l,Y=y)\\
       &=& \sum_k \alpha_k(t)\beta_k(t) 
\eeann
therefore
\bea
\label{eq:ObsvEst}
p(Y=y)&=& \alpha(t)\beta(t) = \gamma(t)
\eea
Equations \ref{eq:Relalpambeta} and \ref{eq:ObsvEst} can be used to estimate $p(X|Y,\theta')$. One can start with a guess of $\alpha(0) = [\frac{1}{q}... \frac{1}{q}]$ and $\beta(N+1)=[1,...1]'$ and then iterate using \ref{eq:FBrec} to get the values of $\alpha(1),...,\alpha(N)$ and $\beta(N),..., \beta(1)$
We can substitute equations \ref{eq:Relalpambeta} and \ref{eq:ObsvEst} to see that 
\bea
\label{Eq:}
p(X|Y,\theta') = \frac{\alpha_k(t,\theta') m_{kl}(t,\theta')\beta_l(t,\theta')}{\gamma(\theta')}
\eea
The expectation maximization algorithm involves two steps. After making an initial guess of parameter $\theta'=\theta_0$ and setting maximum number of iterations $k$ and a tolerance level for the successive estimates $\delta$ we have
\begin{itemize}
\item[i]{Expectation Step\\}
Use the Baum-Welsh forward backward algorithm described above to compute the conditional probabilities $p(X|Y,\theta')$ and  complete likelihood function $Q(\theta | \theta_j)$.
\item[ii]{Maximization Step\\}
Set the new value of $\theta'$
\beann
\theta_{j+1} = \max_\theta Q(\theta | \theta_j)
\eeann
The maximization can be done analytically using Lagrange multipliers or computed numerically. If $\|\theta_{j+1} - \theta_j\| > \delta$ and number of iterations are less $k$ than go to step 1 otherwise set $\theta = \theta_{j+1}$.
\end{itemize}
We generated 200 output symbols using $E =\begin{pmatrix}0.25&0.35& 0.4 \\ 0.15&0.45&0.4\\ 0.25 & 0.25& 0.5 \end{pmatrix}$ and $\epsilon_1 =0.02$ and $\epsilon_2 = 0.03$. The EM algorithm gives $\hat{E} =\begin{pmatrix}0.224 & 0.323 & 0.453 \\ 0.113 & 0.476 & 0.411\\ 0.24 & 0.299 & 0.46 \end{pmatrix}$ and 
$\hat{\epsilon}_1 = 0.048$ and $\hat{\epsilon}_2 = 0.042$. The entropy estimate using Algorithm I with $N=100$, $\hat{E}$, $\hat{\epsilon}_1$  and $\hat{\epsilon}_2$ is $1.51808$ and which is close to the estimate of $1.51715$ using transition matrix E and noise parameters $\epsilon_1$ and $\epsilon_2$. 

%One can also use the Shannon-Mcmillan-Breiman theorem to estimate the entropy rate. The theorem says that for an stationary ergodic stochastic process $\{Y_n\}$ with associated measure $\mu$ the quantity $\frac{\log \mu(Y_1,Y_2,...,Y_n)}{n}$ converges almost surely to the entropy rate $H(\mu)$ \cite{CT1991}. This means that one requires only one sufficiently long sequence of  observations to estimate the entropy rate. Since the function of a Markov process inherits the ergodic properties from its Markov process for ergodic sequences $\{ Y_n \}$ one can estimate the entropy rate using the same forward-backward method and use equation \ref{eq:ObsvEst} to compute $\frac{\log \mu(Y_1,Y_2,...,Y_n)}{n}$ and estimate the entropy rate. Our estimate of the a sequence of 200 output symbols generated by the same transition matrix and noise parameters as in the the case $q=3$ is $1.5047$.

\section{Bounds on the capacity of the Gilbert channel}
\label{sec:Gilbert}
The Gilbert channel \cite{G1960} is a channel with memory which is used to model burst errors. The channel state $S_n$ at time $n$ can be good (G) or bad (B) and the channel transitions between good and bad states according to a Markov chain. When the channel is in a good state the input bit is transmitted without error and when the channel is in a bad state there is a probability of a bit flip is $h$. To model burst errors the channel is modeled so that transition probability from  a good state to a bad state (P) and bad state to good state are small (Q). The output at time n is given by $Y_n = X_n + Z_n$ where $X_n$ is the input at time $n$ and $Z_n$ is the noise and the addition is modulo 2. The noise $Z_n$ will be 0 if the state of the channel is in good state and if the channel is in B state then the noise will be 0 of 1 decided on a coin flip with bias $h$. The noise process $Z_n$ can be looked at as the output of the Z-channel with Markov input $S_n$ and with transition matrix $E=\begin{pmatrix}1-P & P\\ Q & 1-Q\end{pmatrix}$ (see figure \ref{fig:Gilbert Channel}). 
\begin{figure}[ht]
\begin{center}
\includegraphics[height=30 mm,width=70 mm]{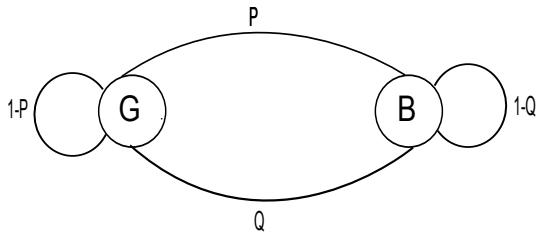}
\caption{The Gilbert channel is a channel with memory. The channel state at time $n$ $S_n$ transitions like a Markov chain between a good state (G) and a bad state (B). The transition probabilities from G state to B state is P and from B to G  is Q. In the good state the channel acts like a perfect channel while in the bad state there is a probability of a bit flip $h$. The output of the Gilbert channel can be written as $Y_n =X_n \oplus Z_n$ where $X_n$ is input process and $Z_n$ is output of the Z-Channel with the channel state Markov process $S_n$  as input. }
\label{fig:Gilbert Channel}
\end{center}
\end{figure}
The capacity of such a finite state channel is defined as
\beann
C = \lim_{n\rightarrow \infty} \frac{1}{n} \max_{p(X^n)} I(X^n;Y^n)
\eeann
where $I(X^n;Y^n)$ is the mutual information between the the input and output sequences. The capacity of finite state Markov channels have been studied in \cite{GV1996,R2006}. We use the results of the previous section to obtain bounds on the capacity of the Gilbert channel. Let $H_N$ be the approximate formula for the entropy rate of the Z-Channel given by theorem \ref{thm:Entropy rate appx} with $q=2$ and $B$ be defined as in equation \ref{eqn:Entropyformula} then we have the following theorem:
\begin{theorem}
Under Assumption 1 the capacity C of the Gilbert channel with channel noise process $\{Z_n\}$ can be upper and lower bounded for all $N \in \mathbb{N}$ as
\beann
1 + H_N(\mu) - B\gamma^{N+1} \leq C \leq 1 + H_N(\mu) + B\gamma^{N+1}
\eeann
\end{theorem}
\begin{proof}
It can be shown that the capacity of the Gilbert channel is
\beann
C = 1 - H(Z)
\eeann
where $H(Z)$ is the entropy rate of the noise process $\{Z_n\}$. Indeed we can write
\beann
I(X^n;Y^n) = \sum_{i=1}^N H(Y_i| Y_1^{i-1}) - H(Y_i|X_i,X_1^{i-1},Y_1^{i-1})
\eeann
Due to the relation $Y_n = X_n \oplus Z_n$ between the input, output and noise we observe that
\beann
H(Y_i|X_i,X_1^{i-1},Y_1^{i-1}) = H(Z_i|Z_1^{i-1})
\eeann
Since Markov channel state process $S_n$ is independent of the input, and the noise process $Z_n$ is the hidden Markov process; $Z_n = \phi(S_n)$ for some function $\phi$ therefore $H(Z_i|Z_1^{i-1})$ is independent of the input distribution $p(X^n)$. Thus
\beann
C &=& \lim_{n\rightarrow \infty} \frac{1}{n} \max_{p(X^n)}\Big( \sum_{i=1}^n H(Y_i| Y_1^{i-1}) - H(Y_i|X_i,X_1^{i-1},Y_1^{i-1})\Big)\\
  &=& \lim_{n\rightarrow \infty} \frac{1}{n} \max_{p(X^n)} \Big(\sum_{i=1}^n H(Y_i| Y_1^{i-1}) - H(Z_i|Z_1^{i-1})\Big)\\
  &=&\lim_{n\rightarrow \infty} \frac{1}{n} \max_{p(X^n)} \sum_{i=1}^n H(Y_i| Y_1^{i-1}) - \lim_{n\rightarrow \infty}\frac{1}{n}\sum_{i=1}^n
 H(Z_i|Z_1^{i-1})
\eeann
Therefore
\bea
\label{eq:GilCap}
C &=&\lim_{n\rightarrow \infty} \frac{1}{n} \max_{p(X^n)} \sum_{i=1}^n H(Y_i| Y_1^{i-1}) - H(Z)
\eea
Now consider
\bea
\label{eq:GilEq1}
&&H(Y_i| Y_1^{i-1}) =\\
\,\,&& -\sum_{y_1,...,y_{i-1}}\big(\sum_y p(Y_i=y|y_1,...,y_{i-1}) \log p(Y_i=y|y_1,...,y_{i-1})\big)p(y_1,...,y_{i-1})\nonumber
\eea
The i.i.d uniform input distribution maximizes $p(Y_i=y|y_1,...,y_{i-1})$. Indeed
\bea
\label{eq:GilEq2}
p(Y_i=y|y_1,...,y_{i-1}) &=& \sum_{s_i} p(Y_i=y|S_i=s_i)p(S_i=s_i|y_1,...,y_{i-1})
\eea
Also,
\beann
p(Y_i=y|S_i=s_i) = \sum_x p(Y_i=y|X_i =x,S_i=s_i)p(X_i =x)
\eeann
Due to the symmetry of the channel $\sum_x p(Y_i=y|X_i =x,S_i=s_i)$ is independent of $Y_i$ therefore for the uniform i.i.d. input distribution the conditional density $p(Y_i=y|S_i=s_i)$ is a constant and is equal to $\frac{1}{2}$. Substituting this in \ref{eq:GilEq2} we get that $p(Y_i=y|y_1,...,y_{i-1}) = \frac{1}{2}$. The quantity 
\beann
\sum_y p(Y_i=y|y_1,...,y_{i-1}) \log p(Y_i=y|y_1,...,y_{i-1})
\eeann
in equation \ref{eq:GilEq1} for the uniform i.i.d. input distribution gets maximized to $2 \frac{1}{2}\log 2 = 1$. Therefore from equation \ref{eq:GilCap} we get 
\beann
C = 1 - H(Z)
\eeann
Using theorem \ref{thm:Entropy rate appx} we can bound the capacity of the Gilbert channel
\beann
1 + H_N(\mu) - B\gamma^{N+1} \leq C \leq 1 + H_N(\mu) + B\gamma^{N+1}
\eeann
\end{proof}
We assume the transition matrix $E = \begin{pmatrix} 0.8 & 0.2 \\ 0.25 &  0.75\end{pmatrix}$ for the channel transitions compute the capacity of the Gilbert channel using different values of $h$ parameter.

\begin{table}[ht]
\begin{center}
\begin{tabular}{|c|c|c|}	
\hline
h          &   C (lower bound)       & C (upper bound)\\
\hline
0.02     &    1.775537282409934   &       1.775537393396272\\
\hline
0.04      &   1.787283765040533   &         1.787284687386383   \\
\hline   
0.06      &   1.797179493884635   &          1.797187010251415   \\
\hline
0.08      &   1.805422838236253 &           1.805483229491273 \\
\hline
0.1        &   1.812015925779448 &          1.812497634488852 \\  
\hline
\end{tabular}
\caption{The upper and lower bounds on the capacity of the Gilbert Channel computed using the entropy rate formula}
\label{table:entropy rate 3}
\end{center}
\end{table}

\bibliographystyle{ieeetr}
\bibliography{ER_comp}
\end{document}